\documentclass[11pt]{article}

\usepackage[margin=1in]{geometry}
\usepackage{amsthm}
\usepackage{mathtools}
\usepackage{amssymb}
\usepackage{enumitem}
\usepackage[hidelinks]{hyperref}

\newtheorem{theorem}{Theorem}
\newtheorem{lemma}[theorem]{Lemma}
\newtheorem{corollary}[theorem]{Corollary}

\newcommand{\E}{\mathbb E}
\newcommand{\1}{\mathbf 1}

\title{A Tight Scale-Locality Bound for Partial Detection in Non-Adaptive Group Testing}
\author{Nader H. Bshouty\\
Technion\\
\texttt{bshouty@cs.technion.ac.il}}
\date{}

\begin{document}

\maketitle

\begin{abstract}
We give a lower bound for randomized non-adaptive group testing when the goal is to find any $\ell$ defective items but the total number $d$ of defectives is unknown. Bshouty and Haddad-Zaknoon proved an upper bound of $O(\ell\log^2 n)$ tests and a lower bound of
\[
  \Omega\!\left(\frac{\ell\log^2 n}{\log \ell+\log\log n}\right).
\]
We prove the matching lower bound. More generally, we show that every randomized non-adaptive algorithm that succeeds with constant probability for every defective set must use
\[
  \Omega\!\left(\ell\log^2(n/\ell)\right)
\]
tests.

The proof is as follows. At a fixed value of $d$, finding $\ell$ defectives requires about $\ell\log(n/d)$ bits of information. On the other hand, one fixed group test is informative only when its size is tuned to the scale of $d$; across all logarithmic scales of $d$, a single test contributes only $O(1)$ bits. Summing over all scales gives the lower bound. 

We also record the matching upper bound $$O\!\left(\ell\log^2(n/\ell)\right),$$obtained by running the known-$d$ algorithm in parallel over dyadic guesses for $d$. Thus the randomized non-adaptive complexity of unknown-$d$ partial detection is $\Theta\!\left(\ell\log^2(n/\ell)\right)$ for constant success probability.
\end{abstract}

\section{Introduction}

Group testing is a classical method for identifying defective items in a large population using pooled tests. Given a ground set $[n]=\{1,\ldots,n\}$ and an unknown defective set $I\subseteq[n]$, a group test is a subset $Q\subseteq[n]$ whose answer is positive if and only if $Q\cap I\ne\emptyset$. The model was introduced by Dorfman in the context of economical blood testing \cite{Dorfman1943}, and has since become a central topic in combinatorics, algorithms, statistics, coding theory, and information theory.

The literature on group testing is extensive. Classical adaptive procedures include the work of Sobel and Groll \cite{SobelGroll1959} and Hwang's generalized binary splitting method \cite{Hwang1972}. Competitive and adaptive variants have been studied in \cite{BarNoyHwangKesslerKutten1994,DuHwangCompetitive1993,DuPark1994,DuXueSunCheng1994,ChengDuXu2014,SchlaghoffTriesch2005,WuChengDu2022}. Non-adaptive group testing is closely connected to pooling designs, superimposed codes, disjunct matrices, cover-free families, and coding theory; see, for example, \cite{BaldingBrunoTorneyKnill1996,DuHwang1993,DuHwangDNA2006a,DuHwangDNA2006b,KautzSingleton1964,DyachkovRykov1982,ChenHwang2007,Furedi1996,Ruszinko1994,PoratRothschild2011,Roth2006}. Randomized group testing and algorithms with optimal or near-optimal query complexity have been investigated in \cite{BshoutyDiabKawarShahla2017,BshoutyHaddadHaddadZaknoon2020,DamaschkeMuhammad2012}. The related problem of estimating the number of defectives has been studied in \cite{FalahatgarJafarpourOrlitskyPichapatiSuresh2016,BshoutyEtAl2018,Bshouty2019}.

Group testing has also found many applications. It appears in DNA library screening and pooling designs \cite{DuHwang1993,DuHwangDNA2006a,DuHwangDNA2006b}, product quality control and sequential screening \cite{Li1962}, image compression \cite{HongLadner2002}, data-stream and frequent-item problems \cite{CormodeMuthukrishnan2005}, and multiaccess communications \cite{Wolf1985}. More recently, neural and machine-learning motivated forms of group testing have been studied in \cite{LiangZou2021}. During the COVID-19 pandemic, group testing was widely discussed as a way to accelerate PCR testing and reduce laboratory costs \cite{BenAmiEtAl2020,CabreraEtAl2020,EisHubingerEtAl2020,GollierGossner2020,MentusRomeoDiPaola2020,ShaniNarkissEtAl2020,YelinEtAl2020,HaddadZaknoon2022}. Related motivations also arise in abnormal-event detection and large-scale image classification, where one may wish to find a small number of positive or abnormal instances among many candidates \cite{KuppusamyBharathi2022,WangSiau2019,XieEtAl2017,LiangZou2021}.

In this paper we study a partial-detection variant of group testing. Instead of identifying all $d=|I|$ defective items, the goal is to output any $\ell$ defective items, where $\ell\le d$. This problem was considered by Ahlswede, Deppe, and Lebedev \cite{AhlswedeDeppeLebedev2012}, who studied the case of finding one defective item and gave bounds for related models. Katona \cite{Katona2011} studied the case of finding at least one excellent element, including deterministic non-adaptive and two-round settings. Gerbner and Vizer \cite{GerbnerVizer2020} generalized these results to deterministic $r$-round algorithms and arbitrary $\ell$.

Partial detection is useful in settings where finding a small number of
confirmed defectives is already enough to trigger a downstream action, such as
triage, screening, or quickly isolating a small infected cluster, and full
identification of all \(d\) defectives may be unnecessary or too costly.

The systematic study of detecting $\ell$ defective items from $d$ defectives was carried out by Bshouty and Haddad-Zaknoon in \cite{BshoutyHaddadZaknoon2023}. They considered adaptive and non-adaptive algorithms, deterministic and randomized algorithms, and both the case where $d$ is known up to a constant factor and the case where no prior information about $d$ is available. In the randomized non-adaptive setting with unknown $d$, they proved an upper bound of $O(\ell\log^2 n)$ tests and a lower bound of
\[
  \Omega\!\left(\frac{\ell\log^2 n}{\log \ell+\log\log n}\right).
\]
Thus their work left a logarithmic gap in the most difficult randomized non-adaptive unknown-$d$ regime.

Our contribution is to close this gap. We prove that every randomized non-adaptive algorithm which, with constant probability, outputs $\ell$ defective items for every defective set must use
\[
  \Omega\!\left(\ell\log^2(n/\ell)\right)
\]
tests. We also record the matching upper bound
\[
  O\!\left(\ell\log^2(n/\ell)\right),
\]
obtained by running the known-$d$ randomized non-adaptive algorithm of \cite{BshoutyHaddadZaknoon2023} in parallel over dyadic guesses for $d$. Therefore the randomized non-adaptive complexity of partial detection with unknown $d$ is
\[
  \Theta\!\left(\ell\log^2(n/\ell)\right)
\]
for constant success probability.

The proof is based on a scale-locality principle. At a fixed value of $d$, outputting $\ell$ defective items requires about $\ell\log(n/d)$ bits of information. On the other hand, a fixed group test is informative only when its size is tuned to the scale of $d$: if the test is too small, it is almost always negative; if it is too large, it is almost always positive. Hence a single test contributes only $O(1)$ bits total over all logarithmic scales of $d$. Summing the fixed-scale information requirements over all scales yields the lower bound. This perspective is related to the scale-locality ideas used in the tight lower bound for non-adaptive estimation of the number of defectives \cite{BshoutyCheungHarcosHatamiOstuni2025}, but the argument here is an entropy direct-sum over all scales rather than a two-distribution total-variation argument.

\subsection{Our Technique}
The proof of the lower bound is based on a scale-locality principle. We consider many possible
values of the number of defectives,
  $d_0,d_1,d_2,\ldots,$
where each value is twice the previous one. First fix one such value~$d$. If the
defective set $I$ is chosen uniformly among all $d$-subsets of $[n]$, then any
one fixed $\ell$-set is contained in $I$ with probability roughly $(d/n)^\ell$.
Therefore, in order to output an $\ell$-set contained in $I$ with constant
probability, the transcript of the algorithm must reveal about
  $\ell\log(n/d)$
bits of information about $I$.

We then look at a single fixed test $Q$. Let $k=|Q|$. If the number of
defectives is $d$, then the expected number of defectives in $Q$ is
  ${kd}/{n}$.
If this quantity is much smaller than $1$, then $Q$ probably contains no
defective item, so the answer to the test is almost always $0$. If this quantity
is much larger than $1$, then $Q$ probably contains at least one defective item,
so the answer is almost always~$1$. In both cases the answer is nearly
predictable and gives very little information. Thus the test can be informative
only for those values of $d$ for which
  ${kd}/{n}$
is close to $1$. Since the possible values of $d$ grow geometrically, this can
happen for only constantly many scales. Hence one fixed test contributes only
$O(1)$ bits total over all logarithmic scales of $d$.

Finally, we sum over all scales. Each scale $d$ requires about
$\ell\log(n/d)$ bits, while each test supplies only $O(1)$ bits in total across
all scales. This gives the lower bound
\[
  \Omega\!\left(\ell\log^2(n/\ell)\right).
\]
This perspective is related to the scale-locality ideas used in the tight lower
bound for non-adaptive estimation of the number of defectives
\cite{BshoutyCheungHarcosHatamiOstuni2025}. The difference is that the
estimation lower bound uses a two-distribution total-variation argument,
whereas the present proof uses an entropy direct-sum over all scales.

The upper bound uses a parallel guessing strategy. If a constant-factor estimate
of $d$ is known, then the randomized non-adaptive algorithm of
\cite{BshoutyHaddadZaknoon2023} finds $\ell$ defective items with the optimal
fixed-scale number of tests. When $d$ is unknown, we run this fixed-scale
algorithm in parallel for all dyadic guesses of $d$, starting from $\ell$ and
going up to $n$. One of these guesses is within a constant factor of the true
value of $d$, so the corresponding copy of the algorithm succeeds. Since all
copies are run in parallel, the resulting algorithm is still non-adaptive. A
direct summation of the fixed-scale costs over the dyadic guesses gives the
matching upper bound $O(\ell\log^2(n/\ell))$.

\section{Model and Main Result}

For a defective set $I\subseteq[n]$ and a test $Q\subseteq[n]$, define
\[
  T_I(Q)=\mathbf 1[Q\cap I\ne\varnothing].
\]
Thus $T_I(Q)=1$ if $Q$ contains at least one defective item, and
$T_I(Q)=0$ otherwise.

A randomized non-adaptive algorithm that detects \(\ell\) defective items
using \(q\) tests first samples a random
seed $S$. After $S=s$ is fixed, the algorithm chooses deterministic tests
\[
  Q_1(s),\dots,Q_q(s)\subseteq[n].
\]
Given a defective set $I$, the answer vector is
\[
  Y(I,s)=\bigl(T_I(Q_1(s)),\dots,T_I(Q_q(s))\bigr)\in\{0,1\}^q.
\]
The algorithm then outputs a set
\[
  L=L(s,Y(I,s))\subseteq[n].
\]
We say that the algorithm {\it succeeds} on $I$ if
\[
  |L|=\ell
  \qquad\text{and}\qquad
  L\subseteq I.
\]
The algorithm has success probability at least $p$ if, for every defective set
$I$ with $|I|\ge \ell$, the probability over the random seed $S$ that it
succeeds on $I$ is at least $p$.

\begin{theorem}[Main theorem]\label{thm:main}
There is an absolute constant $c_0>0$ such that the following holds. Let\footnote{The restriction \(\ell\le n/16\) is used only to focus on the
sparse partial-detection regime. When \(\ell=\Theta(n)\), the expression
\(\ell\log^2(n/\ell)\) is \(\Theta(n)\), and the trivial deterministic
non-adaptive algorithm that tests all singletons gives a linear upper bound.}
 $1\le \ell\le n/16$. Suppose a randomized non-adaptive group testing algorithm succeeds with probability at least $2/3$ on every defective set $I\subseteq[n]$ with $|I|\ge \ell$. Then the number of tests satisfies
\[
  q\ge c_0\,\ell\log^2(n/\ell).
\]
\end{theorem}

\begin{theorem}[Matching upper bound]\label{thm:upper}
There is an absolute constant $C_0>0$ such that the following holds. Let
$1\le \ell\le n/16$. There is a randomized non-adaptive group testing algorithm
which, without knowing $d=|I|$, succeeds with constant probability on every
defective set $I\subseteq[n]$ with $|I|\ge \ell$ and uses at most
\[
  C_0\,\ell\log^2(n/\ell)
\]
tests.
\end{theorem}

\begin{corollary}
For every $1\le \ell\le n/16$, the randomized non-adaptive test complexity of
unknown-$d$ partial detection is
\[
  \Theta\!\left(\ell\log^2(n/\ell)\right)
\]
for a constant success probability. 

In particular, for every fixed $c<1$, if
$\ell\le n^c$, then the complexity is
\[
  \Theta_c(\ell\log^2 n).
\]
\end{corollary}

\section{Lower Bound}

\subsection{Entropy}
For a discrete random variable $X$ taking values in a finite set $\mathcal X$,
its entropy is defined by
\[
  H(X)=\sum_{x\in\mathcal X}\Pr[X=x]\log\frac{1}{\Pr[X=x]},
\]
where terms with $\Pr[X=x]=0$ are omitted. Equivalently,
\[
  H(X)=\mathbb E\!\left[\log\frac{1}{\Pr[X]}\right].
\]
Entropy measures the average number of bits needed to describe the outcome
of $X$.

For two discrete random variables $X$ and $Y$, the conditional entropy of $X$
given $Y$ is
\[
  H(X\mid Y)
  =
  \sum_y \Pr[Y=y]\,H(X\mid Y=y).
\]
Equivalently,
\[
  H(X\mid Y)
  =
  \sum_y \Pr[Y=y]
  \sum_x \Pr[X=x\mid Y=y]\log\frac{1}{\Pr[X=x\mid Y=y]},
\]
where terms with probability $0$ are omitted. The conditional entropy $H(X\mid Y)$ is the average number of additional bits
needed to describe $X$ after the value of $Y$ is known.

We use the following standard facts about entropy; see, for example,
Cover and Thomas~\cite{CoverThomas2006}.
\begin{enumerate}[label=(\roman*)]
  \item\label{MEB} {\bf Maximum entropy bound}: If $X$ takes at most $M$ possible values,
then
\[
  H(X)\le \log M.
\]
Equality holds when $X$ is uniformly distributed over $M$ values.
  \item {\bf Deterministic dependence}: if $X$ is determined by $Y$, that is,
$X=f(Y)$ for some function $f$, then
\[
  H(X\mid Y)=0.
\]

\item {\bf Conditioning reduces entropy}:
  $H(X\mid Y)\le H(X)$.
\item \label{IXZY} {\bf Chain rule for entropy}: $H(X|Y)-H(X|Y,Z)=H(Z|Y)-H(Z|X,Y)$.
  \item\label{SOE} {\bf Subadditivity of entropy}: If $X=(X_1,\dots,X_q)$, then
  \[
    H(X)\le H(X_1)+\cdots+H(X_q).
  \]
  \item If $B$ is a $0$-$1$ random variable with $\Pr[B=1]=p$, then
  \[
    H(B)=h(p):=p\log(1/p)+(1-p)\log(1/(1-p)).
  \]
\end{enumerate}
We also use the standard estimates
\begin{equation}\label{eq:binary-entropy-bound}
  h(p)\le C p\log(1/p)\qquad(0<p\le 1/2)
\end{equation}
for an absolute constant $C$, and the symmetric bound $h(p)=h(1-p)$. 
Indeed, for \(0<p\le 1/2\),
\[
  h(p)=p\log\frac1p+(1-p)\log\frac1{1-p}
  \le p\log\frac1p+2p
  \le C p\log\frac1p,
\]
for an absolute constant \(C\), since \(\log(1/p)\ge \log 2\).

\subsection{One Fixed Scale Requires Many Bits}

First fix the number of defectives $d$. Let $I$ be uniformly random among all $d$-subsets of $[n]$.

We call
  $Z=(S,Y)$
the {\it transcript} of the algorithm. It contains the random seed $S$ and all test
answers $Y$. The output of the algorithm is then a function of the transcript,
 $ L=L(Z)$.

The next lemma says that if an algorithm can output $\ell$ elements of $I$ with constant probability, then its transcript must contain $\Omega(\ell\log(n/d))$ bits about $I$. 
Equivalently, $H(I)$ is the average number of bits needed to describe $I$
before seeing the transcript, $H(I\mid Z)$ is the average number of bits still
needed after seeing $Z$, and $H(I)-H(I\mid Z)=\Omega(\ell\log(n/d))$ is the number of bits about $I$
revealed by $Z$.

\begin{lemma}[Information needed at one scale]\label{lem:one-scale}
Let $\ell\le d\le n/4$. Let $I$ be uniformly random among all $d$-subsets of $[n]$. Let $Z$ be any transcript, and suppose that from $Z$ one can compute an $\ell$-set $L=L(Z)$ such that
\[
  \Pr[L\subseteq I]\ge 2/3.
\]
Then
\[
  H(I)-H(I\mid Z)=\Omega(\ell\log(n/d)).
\]
In words, the transcript $Z$ reveals $\Omega(\ell\log(n/d))$ bits of information about $I$.
\end{lemma}

\begin{proof}
Let $G$ be the event that the output is correct:
  $G=\{L\subseteq I\}$.
By assumption, $\Pr[G]\ge2/3$.

Before seeing any transcript, the defective set $I$ is uniform over $\binom nd$ possibilities, so, by \ref{MEB},
\[
  H(I)=\log\binom nd.
\]
Now imagine that we are only told the output set $L$. If the event $G$ occurs, then $I$ must contain all elements of $L$. Once those $\ell$ elements are fixed, the remaining $d-\ell$ defective items can be chosen in at most
\[
  \binom{n-\ell}{d-\ell}
\]
ways. If $G$ does not occur, we use the trivial bound that $I$ has at most $\binom nd$ possibilities. Also, revealing whether $G$ occurred costs at most one extra bit. Therefore
\[
  H(I\mid L)
  \le 1+\Pr[G]\log\binom{n-\ell}{d-\ell}+(1-\Pr[G])\log\binom nd.
\]
Thus
\begin{align*}
  H(I)-H(I\mid L)
  &\ge \Pr[G]\left(\log\binom nd-\log\binom{n-\ell}{d-\ell}\right)-1 \\
  &= \Pr[G]\log\left(\frac{\binom nd}{\binom{n-\ell}{d-\ell}}\right)-1 \\
  &= \Pr[G]\log\left(\frac{\binom n\ell}{\binom d\ell}\right)-1.
\end{align*}
Since $d\le n/4$ and $\ell\le d$, the ratio $\binom n\ell/\binom d\ell$ is at least a constant multiple of $(n/d)^\ell$. Hence
\[
  H(I)-H(I\mid L)=\Omega(\ell\log(n/d)).
\]
Finally, $L$ is computed from $Z$, so knowing $Z$ can only reduce the uncertainty about $I$ further:
\[
  H(I\mid Z)\le H(I\mid L).
\]
Therefore
\[
  H(I)-H(I\mid Z)\ge H(I)-H(I\mid L)=\Omega(\ell\log(n/d)).
\]
This proves the lemma.
\end{proof}

\subsection{One Test Is Useful at Only Constantly Many Scales}

Now we let $d$ vary over many possible values. Define
\[
  d_j=4\ell\cdot 2^j,
  \qquad j=0,1,\dots,m,
\]
where $m$ is the largest integer with $d_m\le n/4$. Since $d_j$ doubles at every step,
\[
  m=\lfloor\log(n/\ell)\rfloor-4=\Theta(\log(n/\ell)).
\]
For each $j$, let $I_j$ be uniformly random among all $d_j$-subsets of $[n]$.

Fix a single test $Q\subseteq[n]$, and let $k=|Q|$. At scale $j$, define
\[
  B_j(Q)=T_{I_j}(Q)=\1[Q\cap I_j\ne\varnothing].
\]
The following lemma makes precise the idea that a single test has only $O(1)$ total usefulness across all scales.

\begin{lemma}[Scale locality for one test]\label{lem:scale-locality}
For every fixed test $Q\subseteq[n]$,
\[
  \sum_{j=0}^m H(B_j(Q))=O(1).
\]
\end{lemma}

\begin{proof}
Let
  $x_j={k d_j}/{n}$.
This is the expected number of defectives in $Q$ when there are $d_j$ defectives. Since $d_{j+1}=2d_j$, we have
  $x_{j+1}=2x_j$.
So the sequence $x_j$ doubles as $j$ increases.

Let
\[
  p_j=\Pr[B_j(Q)=1]=\Pr[Q\cap I_j\ne\varnothing].
\]
We split into three cases.

\medskip
\noindent\textbf{Small scales: indices $j$ with $x_j\le 1/2$.}
By Markov's inequality,
\[
  p_j=\Pr[|Q\cap I_j|\ge 1]\le \E\left[|Q\cap I_j|\right]=x_j.
\]
Hence, using \eqref{eq:binary-entropy-bound},
\[
  H(B_j(Q))=h(p_j)\le Cx_j\log(1/x_j).
\]
Since the positive numbers $x_j$ double from one scale to the next, the sum of $x_j\log(1/x_j)$ over all $j$ with $x_j\le1/2$ is bounded by an absolute constant.

\medskip
\noindent\textbf{Middle scales: indices $j$ with $1/2<x_j<2$.}
There are at most three such indices $j$, because $x_j$ doubles each time. Each contributes at most one bit, so the total contribution is $O(1)$.

\medskip
\noindent\textbf{Large scales: indices $j$ with $x_j\ge 2$.}
Let
\[
  r_j=\Pr[B_j(Q)=0]=\Pr[Q\cap I_j=\varnothing].
\]
Then
\[
  r_j=\frac{\binom{n-k}{d_j}}{\binom n{d_j}}
  \le \left(1-\frac{k}{n}\right)^{d_j}
  \le e^{-kd_j/n}=e^{-x_j}.
\]
Since $H(B_j(Q))=h(r_j)$ and $r_j\le e^{-x_j}\le e^{-2}<1/2$, again using \eqref{eq:binary-entropy-bound},
\[
  H(B_j(Q))\le Cx_j e^{-x_j}.
\]
The sum of $x_j e^{-x_j}$ over a sequence that doubles at every step and starts at least $2$ is bounded by an absolute constant.

Adding the three ranges gives
\[
  \sum_{j=0}^m H(B_j(Q))=O(1).
\]
\end{proof}

\subsection{Proof of the Main Theorem}

We now prove Theorem~\ref{thm:main}.

Assume the algorithm uses $q$ tests. For each seed $s$, the tests are fixed sets
  $Q_1(s),\dots,Q_q(s)$.
For scale $j$, let
  $Y_j=\bigl(T_{I_j}(Q_1(S)),\dots,T_{I_j}(Q_q(S))\bigr)$
be the answer vector. The algorithm sees both $S$ and $Y_j$, and then outputs
\[
  L_j=L(S,Y_j).
\]
Since the algorithm succeeds with probability at least $2/3$ on every defective set, it also succeeds with probability at least $2/3$ when $I_j$ is uniformly random among all $d_j$-sets. Therefore Lemma~\ref{lem:one-scale}, applied with transcript $Z=(S,Y_j)$, gives
\[
  H(I_j)-H(I_j\mid S,Y_j)=\Omega(\ell\log(n/d_j)).
\]
Since the random seed $S$ is independent of $I_j$, we have $H(I_j)=H(I_j\mid S)$.
By the chain rule for entropy (\ref{IXZY}),
\[
  H(I_j\mid S)-H(I_j\mid S,Y_j)=H(Y_j\mid S)-H(Y_j\mid I_j,S)\le H(Y_j\mid S).
\]
Therefore
  $H(I_j)-H(I_j\mid S,Y_j)\le H(Y_j\mid S)$.
Hence, for every $j$,
  $H(Y_j\mid S)=\Omega(\ell\log(n/d_j))$.
Summing over all scales,
\begin{equation}\label{eq:lower-total-entropy}
  \sum_{j=0}^m H(Y_j\mid S)
  \ge \Omega\!\left(\ell\sum_{j=0}^m\log(n/d_j)\right).
\end{equation}

We next upper-bound the same left-hand side. Fix the seed $S=s$. By subadditivity of entropy~\ref{SOE},
\[
  H(Y_j\mid S=s)
  \le \sum_{i=1}^q H\bigl(T_{I_j}(Q_i(s))\bigr).
\]
Therefore
\begin{align*}
  \sum_{j=0}^m H(Y_j\mid S=s)
  &\le \sum_{j=0}^m\sum_{i=1}^q H\bigl(T_{I_j}(Q_i(s))\bigr) \\
  &= \sum_{i=1}^q\sum_{j=0}^m H\bigl(T_{I_j}(Q_i(s))\bigr).
\end{align*}
For each fixed test $Q_i(s)$, Lemma~\ref{lem:scale-locality} says that the inner sum over $j$ is $O(1)$. Thus
\[
  \sum_{j=0}^m H(Y_j\mid S=s)=O(q).
\]
Averaging over $S$ gives
\begin{equation}\label{eq:upper-total-entropy}
  \sum_{j=0}^m H(Y_j\mid S)=O(q).
\end{equation}
Combining \eqref{eq:lower-total-entropy} and \eqref{eq:upper-total-entropy},
\[
  q\ge \Omega\!\left(\ell\sum_{j=0}^m\log(n/d_j)\right).
\]
Recall that
  $d_j=4\ell\cdot 2^j$.
Therefore
  $\log(n/d_j)=\log({n}/{4\ell})-j$.
By the definition of $m$, we have $m=\lfloor\log(n/\ell)\rfloor-4$. Hence
\[
  \sum_{j=0}^{m}\log(n/d_j)
  =
  \sum_{j=0}^{m}(\log({n}/{4\ell})-j)
  \ge
  \sum_{j=0}^{\lfloor m/2\rfloor}(\log({n}/{4\ell})-j)
  =
  \Omega(\log^2({n}/{4\ell})).
\]
Therefore
\[
  q
  \ge
  \Omega\!\left(
    \ell\sum_{j=0}^{m}\log(n/d_j)
  \right)
  =
  \Omega\!\left(\ell\log^2(n/\ell)\right).
\]
This proves Theorem~\ref{thm:main}.

\section{A Matching Upper Bound}\label{sec:upper}

We now prove the matching upper bound. We use two ingredients from \cite{BshoutyHaddadZaknoon2023}: Theorem 11,
which gives the randomized non-adaptive upper bound when a constant-factor
estimate of $d$ is known, and Lemma~9, which gives a randomized non-adaptive
constant-factor estimator for $d$. We briefly recall the black-box tools from~\cite{BshoutyHaddadZaknoon2023}. The
known-\(D\) algorithm repeatedly isolates a single defective by random sampling
at the appropriate scale \(1/D\), while the estimator probes geometrically many
sampling rates and identifies a constant-factor estimate of \(d\) from the
transition between mostly negative and mostly positive answers.

First, by Lemma~9 in \cite{BshoutyHaddadZaknoon2023}, if a number $D$ is known such that
$d/4\le D\le 4d$, then there is a randomized non-adaptive algorithm that
detects $\ell$ defective items using
  $O\!\left(\ell\log({n}/{d})\right)$
tests for constant success probability. Since $D$ and $d$ differ by at most a
constant factor, this may also be written as
\[
  O\!\left(\ell\log\frac{n}{D}\right).
\]
Second, by Lemma~9 of \cite{BshoutyHaddadZaknoon2023}, there is a
randomized non-adaptive estimator which, using $O(\log(1/\delta)\log n)$
tests, returns a value $\widehat D$ satisfying
  $d/2<\widehat D<2d$
with probability at least $1-\delta$.

The algorithm runs the estimator and, in parallel, runs the known-$D$ algorithm
for all dyadic guesses
\[
  D_i=2^i\ell,
  \qquad
  i=0,1,\dots,\lceil\log(n/\ell)\rceil .
\]
After all test answers are received, the algorithm looks at the estimate
$\widehat D$ and chooses an index $i$ such that $D_i$ is within a factor $2$ of
$\widehat D$. If the estimator is correct, then this $D_i$ is within a constant
factor of the true value of $d$, and therefore the corresponding known-$D_i$
algorithm succeeds with constant probability. Thus the algorithm does not need
to verify which of the parallel outputs is correct; it uses the estimator to
select the appropriate output.

The total number of tests is the sum of the costs over all guesses:
\[
  \sum_{i=0}^{\lceil\log(n/\ell)\rceil}
  O\!\left(\ell\log\frac{n}{2^i\ell}\right)
  =
  O\!\left(\ell\log^2(n/\ell)\right).
\]
The estimator contributes only $O(\log n)$ additional tests for constant success
probability, which is dominated by the above bound. Hence the total number of
tests is
$O\!\left(\ell\log^2(n/\ell)\right)$.

\begin{theorem}[Matching upper bound]\label{thm:upper2}
Let $1\le \ell\le n/16$ and let $0<\delta<1$. There is a randomized
non-adaptive group testing algorithm which, without knowing $d=|I|$, outputs
$\ell$ defective items with probability at least $1-\delta$ for every defective
set $I\subseteq[n]$ satisfying $|I|\ge \ell$, using
\[
  O\!\left((\ell+\log(1/\delta))\log^2(n/\ell)
  +\log(1/\delta)\log n\right)
\]
tests.

In particular, this yields Theorem~\ref{thm:upper}. For a constant success probability, the required number of tests is
$O\!\left(\ell\log^2(n/\ell)\right)$.
\end{theorem}

\begin{proof}
We use two ingredients from \cite{BshoutyHaddadZaknoon2023}. The first is
Theorem 11 there: if a number $D$ is known in advance and
  $d/4\le D\le 4d$,
then there is a randomized non-adaptive algorithm that detects $\ell$ defective
items with probability at least $1-\delta$ using
$O\!\left((\ell+\log(1/\delta))\log(n/d)\right)$
tests. Since $D$ and $d$ differ by at most a constant factor, this can be written
as
$O\!\left((\ell+\log(1/\delta))\log(2n/D)\right)$.

The second ingredient is Lemma 9 of \cite{BshoutyHaddadZaknoon2023}: there is a randomized non-adaptive
estimator which, using $O(\log(1/\delta)\log n)$ tests, returns a value
$\widehat D$ satisfying
  $d/2<\widehat D<2d$
with probability at least $1-\delta$.

We now construct the unknown-$d$ algorithm. Let
\[
  M=\lceil\log(n/\ell)\rceil,
  \qquad
  D_i=\min\{2^i\ell,n\},
  \qquad
  i=0,1,\dots,M.
\]
Since the algorithm is non-adaptive, we cannot first run the estimator, wait for
$\widehat D$, and then choose the tests for the known-$D$ algorithm. Therefore,
before seeing any answers, we run in parallel:
\begin{enumerate}
  \item the estimator, with failure probability at most $\delta/2$;
  \item for every $i=0,1,\dots,M$, the known-$D_i$ algorithm, also with failure
  probability at most $\delta/2$.
\end{enumerate}
All tests are chosen before any answers are observed, so the whole algorithm is
non-adaptive.

After all answers are received, compute the estimate $\widehat D$. Choose an
index $i$ such that $D_i$ is within a factor $2$ of $\widehat D$, and output the
set produced by the corresponding known-$D_i$ algorithm.

If the estimator succeeds, then
  $d/2<\widehat D<2d$.
Since $D_i$ is chosen within a factor $2$ of $\widehat D$, we get
  $d/4<D_i<4d$.
Thus the selected known-$D_i$ algorithm has a valid constant-factor estimate of
$d$, and therefore it succeeds with probability at least $1-\delta/2$. By the
union bound, the estimator succeeds and the selected known-$D_i$ algorithm
succeeds with probability at least $1-\delta$.

It remains to count the tests. The estimator uses
  $O(\log(1/\delta)\log n)$
tests. The known-$D_i$ algorithms use altogether
\begin{align*}
  \sum_{i=0}^{M}
  O\!\left((\ell+\log(1/\delta))\log(2n/D_i)\right)
  &=
  O\!\left((\ell+\log(1/\delta))
  \sum_{i=0}^{M}\log\frac{2n}{2^i\ell}\right)\\
  &=
  O\!\left((\ell+\log(1/\delta))\log^2(n/\ell)\right).
\end{align*}
Therefore the total number of tests is
$O\!\left((\ell+\log(1/\delta))\log^2(n/\ell)
  +\log(1/\delta)\log n\right)$.
For constant $\delta$, this becomes
$O\!\left(\ell\log^2(n/\ell)\right)$.
\end{proof}

Combining Theorems~\ref{thm:main} and~\ref{thm:upper}, the randomized non-adaptive complexity of unknown-$d$ partial detection, for constant success probability, is
$\Theta\!\left(\ell\log^2(n/\ell)\right)$.

\section{Conclusion}

We determined the randomized non-adaptive test complexity of detecting
$\ell$ defective items when the total number $d$ of defectives is unknown.
The main result is the tight bound
$\Theta\!\left(\ell\log^2(n/\ell)\right)$
for constant success probability.

The lower bound follows from a scale-locality argument. At a fixed value of
$d$, any successful algorithm must reveal enough information to identify
$\ell$ defective items, which costs about $\ell\log(n/d)$ bits. However, a
single group test is useful only for a small range of possible values of $d$:
if the test is too small it is almost always negative, and if it is too large
it is almost always positive. Thus each test contributes only $O(1)$ useful
bits over all logarithmic scales of $d$. Summing the information required at
all scales gives the lower bound
$\Omega\!\left(\ell\log^2(n/\ell)\right)$.

The matching upper bound is obtained by running, in parallel, the known-$d$
randomized non-adaptive algorithm for all dyadic guesses of $d$, together with
a non-adaptive estimator for $d$ that is used only after all answers are known
to select the appropriate copy. This keeps the algorithm non-adaptive and gives
the same order of tests as the lower bound.

This closes the gap left by the previous bounds for randomized non-adaptive
partial detection with unknown $d$. 

\paragraph{Open problem.} Develop a scale-locality theory for partial detection under threshold variants, gap-threshold tests, and noisy group-testing models. In particular, determine whether the unknown-\(d\) cost remains an additional logarithmic scale factor, and identify the optimal dependence on the threshold, the gap size, or the noise rate.

\end{document}